\documentclass[runningheads]{llncs}
\usepackage{amsmath}
\usepackage{amssymb}
\usepackage{graphicx}
\usepackage{float}
\usepackage{epstopdf}
\usepackage{booktabs}
\usepackage{subfigure}
\usepackage{array}
\usepackage{makeidx}  
\usepackage{setspace}
\usepackage{multirow}
\usepackage{url}
\usepackage{color}
\usepackage{indentfirst}
\usepackage[lined,ruled,linesnumbered,vlined]{algorithm2e}
\usepackage{soul}

\makeatletter 
\renewcommand{\section}{\@startsection{section}{1}{0mm}
  {-\baselineskip}{0.5\baselineskip}{\bf\leftline}}
\makeatother

\newcommand{\etal}{et al.\xspace}
\newcommand{\ie}{\emph{i.e.,}\xspace}
\newcommand{\eg}{\emph{e.g.,}\xspace}

\usepackage{enumitem}
\setenumerate[1]{itemsep=0pt,partopsep=0pt,parsep=\parskip,topsep=5pt}
\setitemize[1]{itemsep=0pt,partopsep=0pt,parsep=\parskip,topsep=5pt}
\setdescription{itemsep=0pt,partopsep=0pt,parsep=\parskip,topsep=5pt}

\begin{document}
\title{Sorting-based Interactive Regret Minimization}
%
%
%
\author{Jiping Zheng \and Chen Chen}
\authorrunning{J. Zheng and C. Chen}
%
\institute{College of Computer Science and Technology\\
	 Nanjing University of Aeronautics and Astronautics, Nanjing, P.R. China\\
	\email{\{jzh,duplicc\}@nuaa.edu.cn}}

\maketitle

\vspace{-5mm}
\begin{abstract}
As an important tool for multi-criteria decision making in database systems, the regret minimization query is shown to have the merits of top-$k$ and skyline queries: it controls the output size while does not need users to provide any preferences. Existing researches verify that the regret ratio can be much decreased when interaction is available. In this paper, we study how to enhance current interactive regret minimization query by sorting mechanism. Instead of selecting the most favorite point from the displayed points for each interaction round, users sort the displayed data points and send the results to the system. By introducing sorting mechanism, for each round of interaction the utility space explored will be shrunk to some extent. Further the candidate points selection for following rounds of interaction will be narrowed to smaller data spaces thus the number of interaction rounds will be reduced. We propose two effective sorting-based algorithms namely Sorting-Simplex and Sorting-Random to find the maximum utility point based on Simplex method and randomly selection strategy respectively. Experiments on synthetic and real datasets verify our Sorting-Simplex and Sorting-Random algorithms outperform current state-of-art ones.
\vspace{-3mm}
\end{abstract}

\vspace{-8mm}
\section{Introduction}
\label{sec1:intro}

\noindent
To select a small subset to represent the whole dataset is an important functionality for multi-criteria decision making in database systems. Top-$k$ \cite{Ilyas:2008
}, skyline \cite{Borzsony:2001,
Jan:2013} and regret minimization queries \cite{Nanongkai:2010,Nanongkai:2012,Peng:2014,Xie:2019VLDBJ,Xie:2019} are three important tools which were fully explored in the last two decades. Given a \emph{utility} (\emph{preference} or \emph{score} are another two concepts interchangeably used in the literature) function, top-$k$ queries need users to specify their utility functions and return the best $k$ points with the highest utilities. Skyline queries output the points which are not dominated by any other points in the database. Here, domination means two points are comparable. A point $p$ is said to dominate another point $q$ if $p$ is not worse than $q$ in each dimension and $p$ is better than $q$ in at least one dimension. However, both queries suffer from their inherent drawbacks. For skyline queries, the results cannot be foreseen before the whole database is accessed. In addition, the output size of skyline queries will increase rapidly with the dimensionality. Top-$k$ queries ask users to specify their utility functions, but the user may not be able to clearly know what weight each attribute should be, which brings a big challenge to top-$k$ queries. Regret minimization queries return a subset of data points from the database under a required size $k$ that minimizes the maximum regret ratio across all possible utility functions. Here regret ratio of a subset is defined as the relative difference in utilities between the top-1 point in the subset and the top-1 point in the entire database.

The regret minimization query has the merits of both top-$k$ and skyline queries, \ie the output size ($k$) can be controlled while it does not need users to specify the utility functions. Moreover, it has been verified that small regret ratio can be achieved by presenting only a few tuples \cite{Nanongkai:2010}. For example, when there are 2 criteria (dimensions/attributes), 10 points are presented to guarantee a maximum regret ratio of 10\% in the worst case, and the same number of points still make the maximum regret ratio below 35\% for 10 criteria. But the regret ratios shown above cannot make users satisfied. If we want to achieve 1\% maximum regret ratio with 10 criteria, we have to show about 1,000 points to the user \cite{Nanongkai:2010}. Fortunately, it has shown that interaction is much helpful to reduce the user's regret ratio \cite{Nanongkai:2012,Xie:2019}. In \cite{Nanongkai:2012,Xie:2019}, the interaction worked as follows which requires little user effort. When presenting a 
screen of points, the user chooses his/her favorite point. Based on the user's choice, the system modifies the simulated user's utility function and shows another screen of points for next round of interaction until the user has no regret in front of the displayed points or the regret ratio of the user is below a small threshold $\epsilon$. The aim of each interaction round is to approach the user's true utility function which he/she cannot specify. However, the main drawback of existing methods \cite{Nanongkai:2012,Xie:2019} is that they need too many rounds of interaction to achieve a low regret ratio. For example, for a 4-dimensional anti-correlated dataset with 10,000 points generated by the data generator \cite{Borzsony:2001}, the method proposed in \cite{Nanongkai:2012} needs 21 rounds of interaction when displaying 5 points a time to achieve 0.1\% regret ratio. For the algorithms proposed in \cite{Xie:2019}, 9 rounds of interaction are needed when displaying 4 points a time to find the user's favorite point. Too many interaction rounds of the existing methods take too much effort of the user. In this paper, we propose sorting-based interaction mechanism to reduce the rounds of users' interaction. Instead of pointing out the favorite point among the displayed points at each interaction, the user sorts the displayed points according to his/her utility function. As we know that for $s$ data points, the time complexity of choosing the best point is $O(s)$ while the time complexity of sorting $s$ data points is $O(s\log_2 s)$ on average and $O(s)$ in the best case. If $s$ is small, that is, only displaying several points, the time complexities of finding the maximum utility point and sorting have little difference. Thus sorting the displayed points does not increase user's effort. Also, when a user points out the best point, at the same time he/she has browsed all the points which makes him/her easy sort these points, especially in front of only several points. By sorting, our proposed method will need few rounds of interaction because our sorting mechanism can help to shrink the utility function space rapidly. Following is an example to show the pruning power of our sorting-based interactive regret minimization method. Suppose there are 3 points $p_1(10,1)$, $p_2(9,2)$ and $p_3(8,5)$ displayed to the user and the utility space is composed by three utility functions $\{f_1,f_2,f_3\}$ as shown in Table \ref{tab:movitationexample}. The utility is the inner product of point $p$ and utility function $f$, \eg $f_1(p_1)=10\times0.8+1\times0.2=8.2$. Without sorting, when the user points out $p_1$ is his/her favorite point, utility function $f_3$ will not be considered because $f_3(p_3)>f_3(p_1)>f_3(p_2)$ and $f_1$, $f_2$ are both possible user's utility functions. If the user sorts the 3 points with $p_1>p_2>p_3$, utility functions $f_2$, $f_3$ are pruned ($f_2$ is pruned because $f_2(p_1)>f_2(p_3)>f_2(p_2)$). We can see that our sorting based method can faster approach user's actual utility function with fewer rounds of interaction.

\begin{table}
	\setlength{\belowcaptionskip}{-0.2mm}
\centering
\footnotesize
\label{tab:movitationexample}
\caption{Utilities for different utility functions of three points $p_1$, $p_2$ and $p_3$}
\begin{tabular}{|c|c|c|c|c|c|}
  \hline
  \multirow{2}{*}{$p$} &  \multirow{2}{*}{$A_1$} &  \multirow{2}{*}{$A_2$} & $f_1(p)$ & $f_2(p)$ & $f_3(p)$ \\
  & & & $f_1=<0.8,0.2>$ & $f_2=<0.7,0.3>$ & $f_3=<0.6,0.4>$ \\ \hline\hline
  $p_1$ & 10 & 1 & 8.2 & 7.3 & 6.4 \\ \hline
  $p_2$ & 9 & 2 & 7.6 & 6.9 & 6.2 \\ \hline
  $p_3$ & 8 & 5 & 7.4 & 7.1 & 6.8 \\
  \hline
\end{tabular}
\end{table}

\noindent
In summary, the main contributions of this paper are listed as follows.
\begin{itemize}
    \item We propose a sorting-based pruning strategy, which can shrink user's utility space more quickly than existing interactive regret minimization algorithms.
    \item Based on the utility space after pruning, we prune the candidate set by utility hyperplanes to ensure that the displayed points in the next round of interaction are more reasonable and close to the user's favorite point. Two sorting-based interactive regret minimization algorithms, namely Sorting-Random and Sorting-Simplex are proposed based on random and Simplex strategies respectively for displayed points selection.
    \item Extensive experiments on both synthetic and real datasets are conducted to verify efficiency and effectiveness of our sorting-based algorithms which outperform the existing interactive regret minimization algorithms.
\end{itemize}

\noindent
\textbf {Roadmap} Related work is described in Section \ref{sec2:related}. We provide some basic concepts of the regret minimization query as well as some geometric concepts and our interactive framework in Section \ref{sec3:problem}. Our sorting-based technique is introduced in Section \ref{sec4:method}. In Section \ref{sec4:method}, the utility function space pruning strategies via sorting as well as the candidate points selection are detailed. Experimental results on synthetic and real datasets are reported in Section \ref{sec5:exp}. Section \ref{sec6:conclusion} concludes this paper.
%
%
%

\section{Related Work}
\label{sec2:related}

\noindent
Top-$k$ \cite{Ilyas:2008
} and skyline \cite{Borzsony:2001,
Jan:2013} queries are two popular tools for multi-criteria decision making in database systems. However, top-$k$ query requires users to specify their utility functions and it is usually difficult for users to specify their utility functions precisely while skyline query has a potential large output problem which may make users feel overwhelmed.
There are several efforts to control skyline output size, such as $k$-dominant skyline queries \cite{Chan:2006}, threshold-based preferences \cite{Sarma:2011}, top-$k$ representative skyline \cite{lin:2007}, distance-based representative skyline \cite{Tao:2009} etc. To bridge the gap of top-$k$ query for specifying accurate utility functions and skyline query for outputting too many results, regret-based $k$ representative query which was proposed by Nanongkai et al. \cite{Nanongkai:2010} tries to output a specified size \eg $k$ while minimizing user's maximum regret ratio.

Following researches are along with the regret minimization query \cite{Nanongkai:2010} from various aspects. Peng \etal \cite{Peng:2014} introduce the concept of \textit{happy points} in which the final $k$ points included to speed up the query process. Approximate solutions in polynomial time with any user-specified accurate thresholds are proposed in \cite{Asudeh:2017,Agarwal:2017} or with asymptotically optimal regret ratio in \cite{Xie:2018}. \cite{Nanongkai:2012,Xie:2019} investigate how interaction is helpful to decrease users' regret ratios. Chester \etal \cite{Chester:2014} relax regret minimization queries from top-1 regret minimization set to top-$k$ minimization set which they call $k$-RMS query. Further, coreset based algorithms \cite{Agarwal:2017,Kumar:2018,Cao:2017} or hitting set based algorithms \cite{Agarwal:2017,Kumar:2018} are developed to solve the $k$-RMS problem efficiently and effectively. Faulkner et al. \cite{Faulkner:2015} and Qi et al. \cite{Qi:2018} extend linear utility functions used in \cite{Nanongkai:2010,Nanongkai:2012,Peng:2014,Xie:2018,Xie:2019} to \textsc{Convex}, \textsc{Concave} and \textsc{CES} utility functions and multiplicative utility functions respectively for $k$-regret queries. Zeighami and Wong \cite{Zeighami:2016} propose the metric of average regret ratio to measure user's satisfaction against output results and further developed efficient algorithms to solve it \cite{Zeighami:2019}.

From the variants of the regret minimization query listed above, the most related to our research is \cite{Nanongkai:2012} and \cite{Xie:2019}. Nanongkai \etal \cite{Nanongkai:2012} first enhance traditional regret minimization sets by user interaction. At each round of interaction, the user is presented a screen of artificial data points which have the great possibility to attract user's attentions for next interaction. Then the system asks the user to choose his/her favorite point. Based on the user's choice, the system learns user's utility function implicitly. With limited number of interaction rounds, the user may find his/her favorite point or the point within a specified small regret ratio. Xie \etal \cite{Xie:2019} argue that displaying fake points to users \cite{Nanongkai:2012} makes users disappointed for they are not indeed inside the database. Also the number of interaction rounds for the proposed method in \cite{Nanongkai:2012} is a little large. In this paper, we follow the paradigm of interactive regret minimization. Instead of pointing out the most favorite point at each round of interaction, we sort the displayed data points and fully exploit the pairwise relationship among displayed points of each interaction to narrow the utility space. Thus our proposed sorting-based interactive regret minimization which needs much less rounds of interaction than existing approaches \cite{Nanongkai:2012,Xie:2019}.

\section{Preliminaries}
\label{sec3:problem}

\noindent
Before we give our interaction framework (Section \ref{subsec:problemdefinition}), we first introduce some basic concepts for the regret minimization query (Section \ref{subsec:regretquery}). Then useful geometric concepts such as boundary points, convex hull and conical hull frame etc. are listed in Section \ref{subsec:geometricconcepts}. 
\vspace{-3mm}
\subsection{Regret Minimization Query}
\label{subsec:regretquery}

\noindent
Let $D$ be a set of $n$ $d$-dimensional points over positive real values. For each point $p\in D$, the value on the $i$th dimension is represented as $p[i]$. Related concepts of the regret minimization query are formally introduced as follows \cite{Nanongkai:2010}.

\noindent
\textbf{Utility function.} A user utility function $f$ is a mapping $f$: $\mathbb{R}_+^d \rightarrow \mathbb{R}_+$.
Given a utility function $f$, the utility of a data point $p$ is denoted as $f(p)$, which shows how satisfied the user is with the data point $p$.

Obviously, there are many kinds of utility functions, such as convex, concave, constant elasticity of substitution (CES) \cite{Faulkner:2015} and multiplicative \cite{Qi:2018} etc. In this paper, we focus on linear utility functions which are very popular to model users' preferences \cite{Nanongkai:2010,Nanongkai:2012,Peng:2014,Xie:2018,Xie:2019}.

\noindent
\textbf{Linear utility function.} Assume there are some nonnegative real values $\{v_1,$ $v_2, \cdots, v_d\}$, where $v_i$ denotes the user's preference for the $i$th dimension. Then a linear utility function can be represented by these nonnegative reals and $f(p)=\sum_{i=1}^d v_i\cdot p[i]$. A linear utility function can also be expressed by a vector\footnote{In the following, we use utility function and utility vector interchangeably.}, \ie $v=<v_1, v_2, ..., v_d>$, so the utility of point $p$ can be expressed by the dot product of $v$ and $p$, \ie $f(p)=v\cdot p$.

\noindent
\textbf{Regret ratio.} Given a dataset $D$, a subset $S$ of $D$ and a linear utility function $f$, the regret ratio of $S$, represented by $rr_D(S, f)$, is defined as
	\[
	rr_D(S,f)=1-\frac{\max_{p\in S} f(p)}{\max_{p\in D}f(p)}
	\]	

\noindent
Since $S$ is a subset of $D$, given a utility function $f$, it is obvious that $\max_{p\in S}f(p)\leq\max_{p\in D}f(p)$ and the $rr_D(S,f)$ falls in the range $[0,1]$. The user along with utility function $f$ will be satisfied if the regret ratio approaches 0 because the maximum utility of $S$ is close to the maximum utility of $D$.

\noindent
\textbf{Maximum regret ratio.} Given a dataset $D$, a subset $S$ of $D$ and a class of utility functions $\mathcal{F}$. The maximum regret ratio of $S$, represented by $rr_D (S, \mathcal{F})$, is defined as 	
	\[
	rr_D (S,\mathcal{F})=\sup_{f\in \mathcal{F}}rr_D(S,f)=\sup_{f\in \mathcal{F}}\left(1-\frac{\max_{p\in S} f(p)}{\max_{p\in D} f(p)}\right)
	\]

To better understand above concepts, we present a concrete car-selling example for illustration. Consider a car database containing 5 cars with two attributes namely miles per gallon (MPG) and horse power (HP) whose values are normalized as shown in Table \ref{tab:carexample}. Let a linear utility function $f=<0.7,0.3>$. The utilities of 5 cars under the utility function $f$ are shown in the 4th column of Table \ref{tab:carexample}. We can see that the point with the maximum utility 0.69 is $p_5$. If we select $p_2,p_4$ as the result set, that is, $S=\{p_2,p_4\}$, we can obtain the regret ratio $rr_D(S,f)=1-\frac{\max_{p\in S}f(p)}{\max_{p\in D}f(p)}=1-\frac{0.61}{0.69}=11.6\%$.

\begin{table}[!htb]
	\setlength{\belowcaptionskip}{-0.2mm}
 \centering
  \caption{Car database and the utilities under $f$}
  \label{tab:carexample}
 \footnotesize
  \begin{tabular}{|c|c|c|c|}\hline
    Car      &MPG      &HP     &$f(p)$    \\ \hline\hline
    $p_1$    & 0.4       & 0.8      & 0.52            \\ \hline
    $p_2$    & 0.6       & 0.5      & 0.57             \\ \hline
    $p_3$    & 0.3       & 0.6      & 0.39              \\ \hline
    $p_4$    & 0.7       & 0.4      & 0.61                \\ \hline
    $p_5$    & 0.9       & 0.2      & 0.69                \\ \hline
  \end{tabular}
\end{table}
\vspace{-5mm}
\subsection{Geometric Concepts for Interactive Regret Minimization}
\label{subsec:geometricconcepts}

\noindent
Similar to \cite{Xie:2019}, interesting geometric properties can be exploited to prune the utility space and compute the maximum regret ratio of a given subset more easily. Before we define our problem, we provide useful geometric concepts for our interactive regret minimization.

\noindent
\textbf{Boundary point.}  Given a $d$-dimensional dataset $D$ of $n$ points, a point $p \in D$ is said to be an $i$th  ($i \in [1,d]$) dimension boundary point of $D$ if $p[i]$ is the largest value among all points in $D$ in $i$th dimension. Consider our example in Fig. \ref{fig:convex-hull} showing a set $D$ of 7 data points, namely $p_1,p_2,\ldots,p_7$ in a 2-dimensional space with two dimensions $A_1$, $A_2$. We can see that $p_5$, $p_1$ are the boundary points corresponding to $A_1$, $A_2$ respectively. When the values of all the points in each dimension are normalized to $[0,1]$ and let $b_i[j]=1$ if $j=i$, and $b_i[j]=0$ if $j\neq i$ where $i,j=1,...,d$, we say that $b_i$s are boundary points of $D\cup\{b_1,b_2,...,b_d\}$.
\begin{figure}[H]
	\setlength{\abovecaptionskip}{-1.5mm}
	\setlength{\belowcaptionskip}{-1mm}
	\begin{minipage}[t]{0.5\textwidth}
		\centering
		\includegraphics[width=0.5\textwidth]{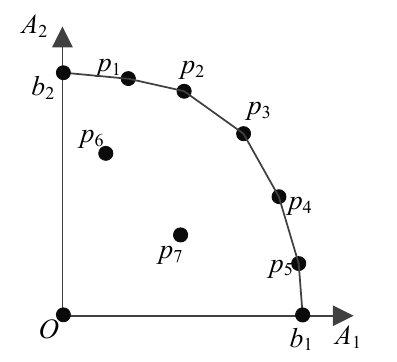}
		\caption{Convex hull with boundary points}
		\label{fig:convex-hull}
	\end{minipage}
	\begin{minipage}[t]{0.5\textwidth}
		\centering
		\includegraphics[width=0.5\textwidth]{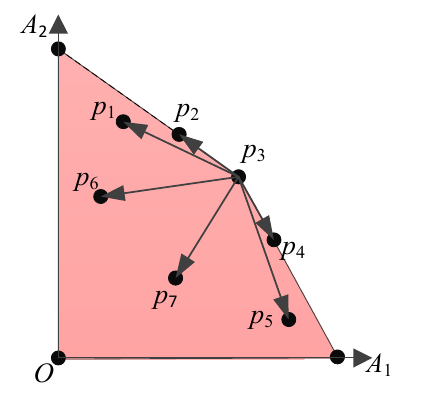}
		\caption{Conical hull}
		\label{fig:conical-hull}
	\end{minipage}
\end{figure}
\noindent
Next important geometric concept is convex hull in which points have great possibility to be included in the result set of the regret minimization query \cite{Peng:2014,Asudeh:2017}.

\noindent
\textbf{Convex hull.}
In geometry, the convex hull of $D$, denoted by $Conv(D)$, is the smallest convex set containing $D$. A point $p$ in $D$ is a vertex of $Conv(D)$ if $p\notin Conv(D/\{p\})$. 
In 2-dimensional space, let $O=(0,0)$ be the origin and $b_1$, $b_2$ are two boundary points of $D\cup\{b_1,b_2\}$. Fig. \ref{fig:convex-hull} shows the convex hull of points set $D\cup\{b_1,b_2,O\}$, denoted as $Conv(D\cup\{b_1,b_2,O\})$. Note that for any linear utility function, the point in $D$ with the maximum utility must be a vertex of $Conv(D\cup\{b_1,b_2,O\})$. Here, we say a point in $Conv(D\cup\{b_1,b_2,O\})$ to be a vertex of the hull.

Although the maximum utility point for each linear utility function lies in the convex hull of $D\cup\{b_1,b_2,O\}$, investigating each point in the convex hull to find the point with maximum utility is too time-consuming because the number of points in a convex hull is usually very large. Even in 2-dimensional case, the convex hull can be as large as $O(n^{1/3})$ and for a database with 5 dimensions, the convex hull can often be as large as $O(n)$ \cite{Asudeh:2017}. Thus instead of the convex hull, the concept of conical hull frame \cite{Dula:1998} helps to find a small subset of the convex hull for the maximum utility point investigation. Following are three geometric concepts to find this kind of subset.

\noindent
\textbf{Conical hull.} Given a vertex $p$ in $Conv(D)$, we let vector set $V$ = $\{q-p|\forall q\in D/\{p\}\}$. The conical hull of a point $p$ \emph{w.r.t.} $V$ is defined to be $C_{p,V}=\{q\in \mathbb{R}^d|(q-p)=\sum_{v_i\in V}w_i\cdot v_i\}$ where $w_i\geqslant0$ \cite{Dula:1998} and the conical hull $C_{p,V}$ is also a convex cone with apex $p$ \cite{Rockafellar:2015}. Fig. \ref{fig:conical-hull} shows an example of conical hull in 2-dimensional space. In Fig. \ref{fig:conical-hull}, the conical hull of point $p_3$ is $\{p_2-p_3,p_1-p_3,p_6- p_3,p_7-p_3,p_5-p_3,p_4-p_3\}$ which is the shaded region in Fig. \ref{fig:conical-hull}.

\noindent
\textbf{Conical hull frame.} A set $V_F \subseteq V$ is defined to be a conical hull frame of a vector set $V$ if $V_F$ \emph{w.r.t.} a point $p$ is the minimal subset of $V$ such that $V_F$ and $V$ have the same conical hull of $p$, \ie $C_{p,V}$ = $C_{p,V_F}$. It is obvious that for each vector $v \in V_F$, we have $v \notin C_{p,V/\{v\}}$. In Fig. \ref{fig:conical-hull}, for point $p_3$ and vector set $V =\{p_i-p_3|\forall p_i\in D/\{p_3\}\}$, the conical hull frame of $V$ \emph{w.r.t.} $p_3$ is $V_F=\{p_2-p_3, p_4-p_3\}$ which is the frame of $V$ since it is the minimal subset of $V$ such that $C_{p_3,V} = C_{p_3,V_F}$.

\noindent
\textbf{Neighbouring vertex.} As the name suggests, the neighbouring vertex set $N_p$ of a point $p$ is composed of the neighbors of $p$ in the convex hull. For the example in Fig. \ref{fig:conical-hull}, the neighbouring vertexes of $p_3$ in $Conv(D)$ are $p_2$ and $p_4$. For a utility function $f$ and a point $p\in Conv(D)$, either $p$ is the maximum utility point to $f$ or there exists a vertex in $N_p$ whose utility is larger than that of $p$ \cite{Xie:2019}. Based on this, if $p$ is not the maximum utility point, we can find a better one in $N_p$. Intuitively, $N_p$ can be selected after the computation of the whole convex hull. As mentioned above, computing the whole convex hull is time-consuming. Fortunately, \cite{Xie:2019} shows that the conical hull frame of $V$ is close to $N_p$, \ie $q\in N_p$ if and only if $q-p\in V_F$ which makes it efficient to be calculated.

\vspace{-3mm}
\subsection{Sorting-based Interaction}
\label{subsec:problemdefinition}

\noindent
Our sorting-based interaction framework works as follows. Initially, the system interacts with a user \emph{w.r.t.} an unknown utility function, displaying $s$ points for the user to sort. We restrict $s$ to be a small integer not bigger than 10 to alleviate the burden of sorting. After the user's feedback, \ie returning the sorting list to the system, we shrink the utility space which the user's utility function may be in and prune the non-maximum utility points in the candidate set. After certain rounds of interaction like this, the system returns a point with the regret ratio below a predefined value $\epsilon$. Here, $\epsilon$ ranges from 0\% to 100\%. If $\epsilon=0$, it means that the user has no regret on the point returned by the system.

The main problem is the rounds of the user's interaction needed for our interaction framework. Comparing to the existing methods which only select the favorite point at each interaction round, by introducing sorting mechanism we can fully exploit the information the user has provided and quickly find the favorite point in the database $D$. Next section we show how sorting can help to reduce rounds of interaction for regret minimization queries.

\section{Sorting-based Interaction for Regret Minimization Queries}
\label{sec4:method}

\noindent
In this section, we first illustrate sorting is helpful to shrink the utility space which the user's unknown utility function falls in. Then, we provide the strategies to select the points for next round of interaction.

\subsection{Utility Space Shrinking via Sorting}

\noindent
In each iteration, when the displayed $s$ points are sorted and returned to the system, the system will shrink the utility function space $\mathcal{F}$ to some extent. We first define the concept of utility hyperplane then illustrate our utility space pruning procedure. Given two points $p$ and $q$, we define a utility hyperplane, denoted by $h_{p,q}$, to be the hyperplane passing through the origin $O$ with its normal in the same direction as $p-q$. The hyperplane $h_{p,q}$ partitions the space $\mathbb{R}^d$ into two halves. The half space above $h_{p,q}$ is denoted by $h_{p,q}^+$ and the half space below $h_{p,q}$ is denoted by $h_{p,q}^-$. The following lemma from \cite{Xie:2019} shows how we can shrink $\mathcal{F}$ to be a smaller space based on utility hyperplane.

\begin{lemma}
\label{lemma:shrinking}
Given utility space $\mathcal{F}$ and two points $p$ and $q$, if a user prefers $p$ to $q$, the user's utility function $f$ must be in $h_{p,q}^+\bigcap \mathcal{F}$.
\end{lemma}
We can find that the half space $h_{p,q}^+$ represents the range of all possible utility functions for $p$ is prior to $q$. For example in Fig. \ref{fig:shrinking}, the system presents three points in 3-dimensional space to the user, $p=(\frac{1}{2},0,\frac{1}{2})$, $q=(0,\frac{1}{2},\frac{1}{2})$ and $r=(\frac{1}{2},\frac{1}{2},0)$, the user sorts $p, q, r$ based on his/her unknown utility function. The region of $\triangle ABC$ represents all possible values of utility functions, $\sum_{i=1}^{d}f[i]=1$. Sorting information can be fully exploited as follows.
\begin{itemize}
  \item According to $f(p)>f(q)$, the utility hyperplane $Om_1n_1p_1$ (the blue rectangle in Fig. \ref{fig:shrinking}(a)) is constructed. The part where the hyperplane intersects with the $\triangle ABC$ is a straight line $Ar$, where the region of $\triangle ABr$ contains all possible utility functions that satisfy $f(p)>f(q)$, and the region of $\triangle ACr$ contains all possible utility functions that satisfy $f(p)<f(q)$. So the utility space $\triangle ABr$ is reserved.
  \item For $f(p)>f(r)$ and $f(r)>f(q)$, similar to the above analysis, only the regions of $\triangle ABq$ and  $\triangle BCp$ are reserved.
\end{itemize}
 After this interaction round, the utility space $\mathcal{F}$ containing the user's utility function shrinks from $\triangle ABC$ to $\triangle Bpt=\triangle ABC\cap\triangle ABr\cap\triangle ABq\cap\triangle BCp$ (Fig. \ref{fig:ShrunkUS}(a)). As a contrast, if only selecting the favorite point $p$ at this round, \ie without (WO) sorting, it implies $f(p)>f(q)$ and $f(p)>f(r)$, the utility space only shrinks from $\triangle ABC$ to $\triangle AtB=\triangle ABC\cap\triangle ABr\cap\triangle ABq$ as Fig. \ref{fig:ShrunkUS}(b) shows. It is obvious that the shrunk utility space $\mathcal{F}_{sorting}$ belongs to the shrunk utility space $\mathcal{F}_{nosorting}$ without sorting, \ie $\mathcal{F}_{sorting}\subseteq\mathcal{F}_{nosorting}$. The idea shown here can be naturally extended to high dimensional data space, thus our sorting-based interactive regret minimization is superior to existing methods \cite{Nanongkai:2012,Xie:2019}.
\begin{figure}[htbp]
	\setlength{\abovecaptionskip}{-1.5mm}
	\setlength{\belowcaptionskip}{-1mm}
    \centering
    \subfigure[$f(p)>f(q)$]{\includegraphics[width=0.22\textwidth]{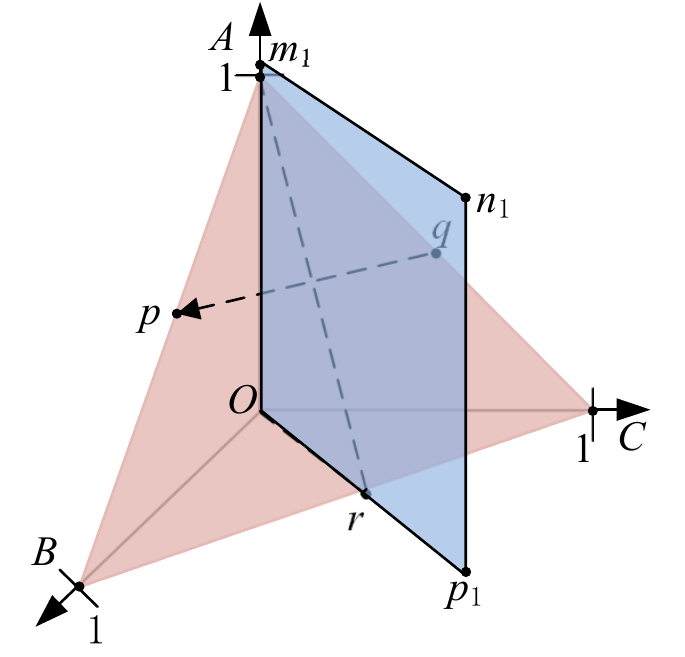}}
    \subfigure[$f(p)>f(r)$]{\includegraphics[width=0.22\textwidth]{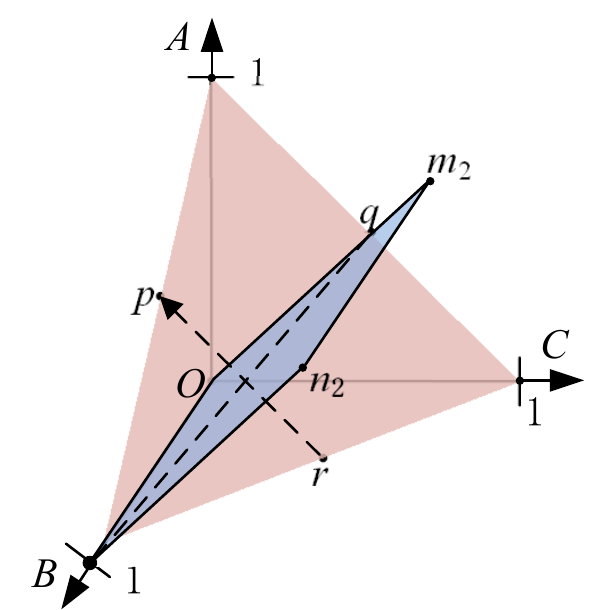}}
    \subfigure[$f(r)>f(q)$]{\includegraphics[width=0.22\textwidth]{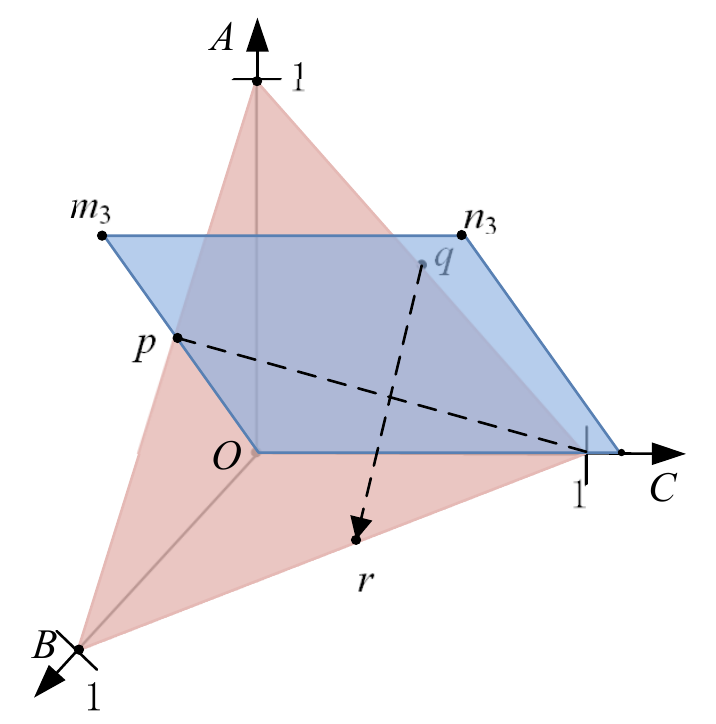}}
    \caption{Utility Space Shrinking via Sorting}
    \label{fig:shrinking}
\end{figure}

\begin{figure}[htbp]
	\setlength{\abovecaptionskip}{-1.5mm}
	\setlength{\belowcaptionskip}{-1mm}
    \centering
    \subfigure[with Sorting]{\includegraphics[width=0.19\textwidth]{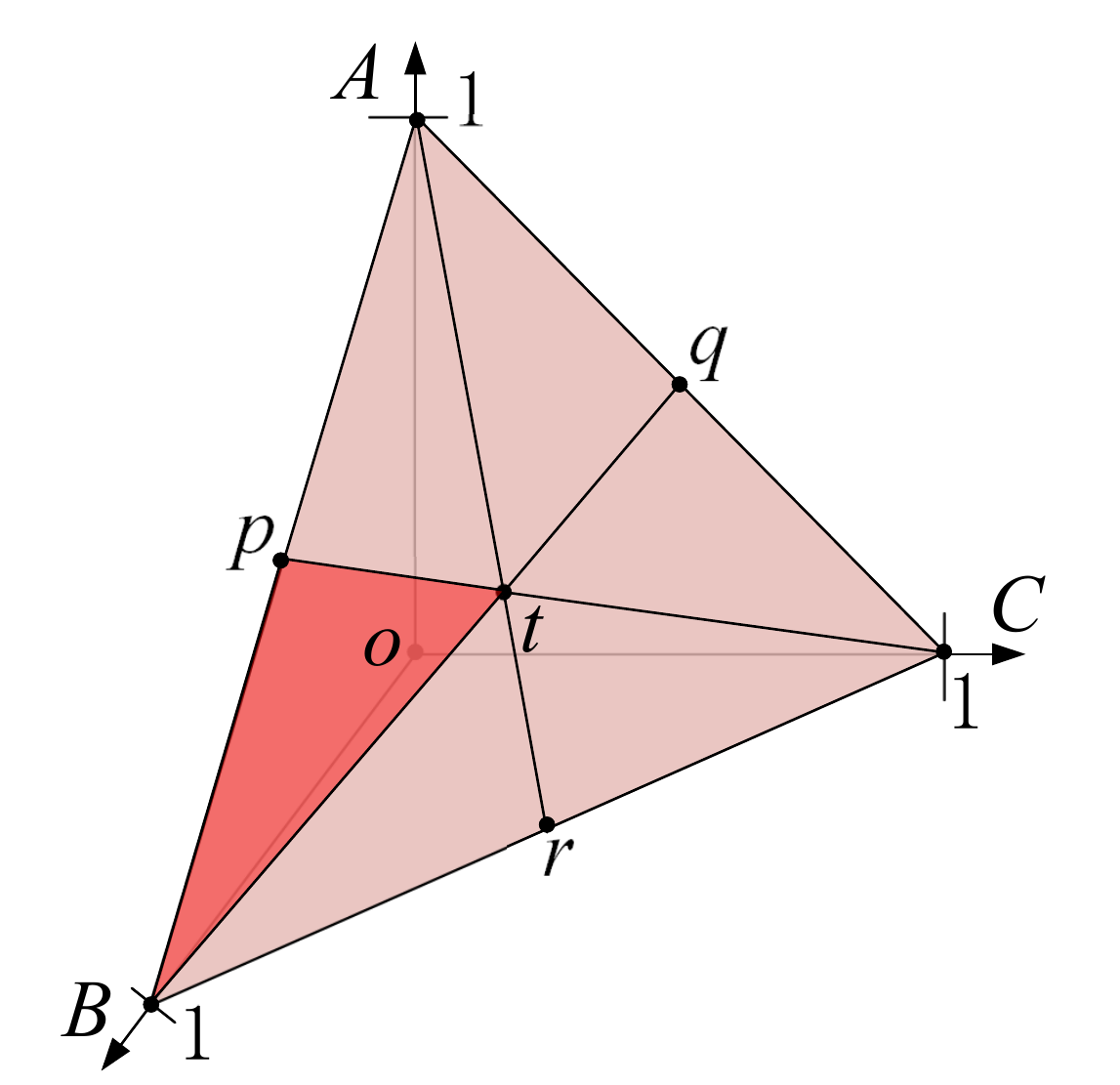}}
    \subfigure[WO Sorting]{\includegraphics[width=0.19\textwidth]{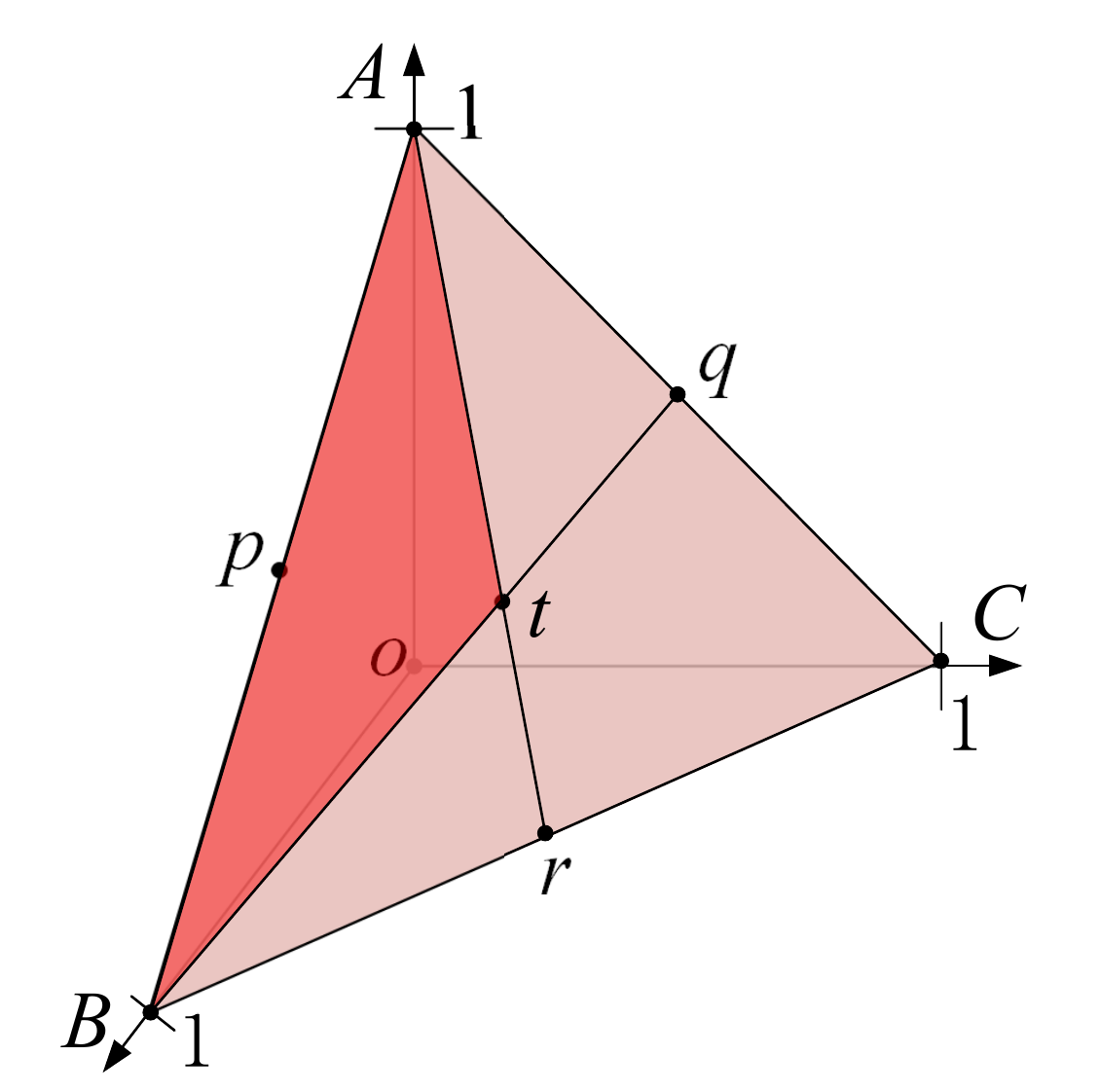}}
    \caption{Shrunk Utility Space}
    \label{fig:ShrunkUS}
\end{figure}

\noindent
When the system gets the feedback from the user, it prepares to select specified $s$ points presenting to the user for next round of interaction. Since the utility space has been shrunk to a smaller space, correspondingly, the candidate set for selecting points displayed to the user will also be reduced. In the literature, the skyline of the dataset $D$ is initially regarded as the candidate set $C$. We provide a strategy named utility hyperplane based candidate set pruning to reduce the size of the candidate set $C$ by removing non-maximum utility points in $C$.
From Lemma \ref{lemma:shrinking}, if a utility function $f$ falls in $h_{p,q}^+\bigcap \mathcal{F}$, we can say that the user prefers $p$ to $q$. That is, we can safely prune $q$ if there is a $p$ in $C$ when $f$ is the user's utility function from $h_{p,q}^+\bigcap \mathcal{F}$. We summarize utility hyperplane based candidate set pruning by Lemma \ref{lemma:hyperplane}.

\begin{lemma}
\label{lemma:hyperplane}
Given the utility space $\mathcal{F}$, a point $q$ can be pruned from $C$ if there exists a point $p$ in $C$ such that $h_{q,p}^+\bigcap\mathcal{F}=\phi$.
\end{lemma}

\subsection{Displayed Points Selection}

\noindent
There are two strategies to select $s$ points from $C$ presenting to the user, namely \emph{random} and \emph{Simplex}. The idea of random strategy is to randomly select $s$ points from the candidate set $C$ to the user. The Simplex approach, based on the conical hull frame, according to the user's favorite point $p$ in the previous interaction round, uses Simplex method to pick neighbouring points of $p$ in the convex hull and displays them to the user. The idea of Simplex strategy is borrowed from the Simplex method for Linear Programming (LP) problems \cite{Dula:1998,Xie:2019}. Note that the maximum utility point must be a vertex in $Conv(D)$. It is time-saving to interactively check if there is a vertex in $Conv(D)$ with a higher utility than $p$ by displaying $p$ and at most $s-1$ neighboring vertices in $Conv(D)$ represented as $N_p$ to the user at each round. So we just present points in $C$ which are also vertices in $N_p$, \ie the vertex set $\{p \in C \bigcap N_p\}$ to the user. Instead of obtaining $N_p$ from $Conv(D)$ which is time-consuming to compute it for high dimensional dataset, we compute $N_p$ by $p$'s conical hull frame. Similar to \cite{Xie:2019}, we use the algorithm in \cite{Dula:1998} to compute the conical hull frame. Based on above analysis, we provide our sorting-based interactive regret minimization algorithm Sorting-Simplex as shown in Algorithm 1. Also, we propose our Sorting-Random algorithm using random points selection strategy instead of Simplex method (in line 9).

\begin{algorithm}[!htb]
    \scriptsize
	\caption{Sorting-Simplex Algorithm}
	\label{alg:1}
    \KwIn{dataset $D$, a regret ratio $\epsilon$, displayed points per interaction $s$, an unknown utility vector $f$, displayed point set $T$}
    \KwOut{a point $p$ in $D$ with $rr_Df(p)\leqslant \epsilon$}
    \BlankLine	
	Initially, $\mathcal{F} \leftarrow {f \in \mathbb{R}_+^d| \sum_{i=1}^d f[i] =1}$\;
    $C \leftarrow \text{the set of all skyline points}\in D$\;
    $p \leftarrow \text{a vertex of }  Conv(D)$\;
         \While{$||\mathcal{F}||_1>\frac{\epsilon}{2d}$ and $|C|>1$}
         {   
            $T\leftarrow$ display $p$ and $s-1$ points in $N_p \cap C$\; 
            $L\leftarrow \text{sort the points in $T$ with }f(L[1])>f(L[2])>...>f(L[s])$\;
              \If{$L[1]\neq p$}
        {
         $p\leftarrow L[1]$\;
         Use Simplex method to choose the neighboring vertices of $p$ in $Conv(D)$ to $N_p$\;
        }
       \For{$i=0$,$i<s$,$i++$}
       {
           \For{$j = 0$,$j < s$,$j++$}{
               \If{$i \neq j \ and  \ i < j$}{
                    $\mathcal{F}=\mathcal{F}\bigcap h_{L[i],L[j]}^+$\;
                   }
          }
       }
       \For{$i = 0, i < | C |,i++$}
           {
               \For{$j = 0, j < | C |,j++$}
                   {
                      \If{$i \neq j \ and \ i < j$}
                           {
                               \If{$h_{C[i],C[j]}^+ \bigcap \mathcal{F} = {\O}$}
                                   {
                                           remove C[$i$]\;
                                   }
                               \If{$h_{C[i],C[j]}^- \bigcap \mathcal{F} = {\O}$}
                                   {
                                           remove C[$j$]\;
                                   }
                           }
                   }
           }
        }
\textbf{return}  $p=\arg\max_{q\in C}f\cdot q$ where $f\in \mathcal{F}$;
\end{algorithm}

In Algorithm 1, we first initialize the candidate set $C$ to be the skyline of $D$ and the utility space $\mathcal{F}$ to be the whole linear utility space (lines 1-2). Then we choose a point in the convex hull of $D$ (line 3). If not satisfying the stop condition (below a small regret ratio with $||\mathcal{F}||_1>\frac{\epsilon}{2d}$ \cite{Xie:2019} or $C$ has only one point, line 4), Algorithm 1 will choose $s$ points presented to the user (line 5) and the user sorts the $s$ points in descending order of their utilities (line 6). The system exploits Simplex method to obtain the neighbouring vertexes in the convex hull (line 7-9). Then the system shrinks the utility space $\mathcal{F}$ with $C_s^2$ utility hyperplanes (lines 10-13) and reduces the size of the candidate set $C$ with utility hyperplane pruning (line 14-20). At length, the system returns the user's favorite point or the point with the regret ratio no larger than $\epsilon$.

Following we present the lower bound of the number of interaction rounds needed to return the user's favorite point.

\begin{theorem}
	For any $d$-dimensional dataset, there is an algorithm that needs $\Omega(log_{C_s^2}n)$ rounds of interaction to determine the user's favorite point.
\end{theorem}
\begin{proof}
The step of determining the user's favorite point in the interaction can be simulated in the form of a tree. Consider an $s$-ary tree with height $r$, $r$ representing the rounds of interaction, each leaf node of the $s$-ary tree representing the data point in $D$. If we show the $s$ points in each interaction round, the user sorts the $s$ points, and we can get $C_{s}^2$ comparison information, similar to in the case of no comparison showing 2 points a time for $C_s^2$ rounds. Since it is a $s$-ary tree with $n$ leaves, the height of the tree is $\Omega(log_{C_s^2}n)$. In other words, any algorithm needs $\Omega(log_{C_s^2}n)$ rounds of interaction to identify the maximum utility point in the worst case.
\end{proof}
\begin{table}[!htb]
	\setlength{\belowcaptionskip}{-0.2mm}
 \centering
  \caption{UH-Simplex example with utility function $f=<0.3,0.3,0.2,0.2>$}
  	\setlength{\tabcolsep}{0.3mm}
  \label{tab:UH-Simplex}
 \footnotesize
  \begin{tabular}{|c|l|c|c|c|c|c|c|c|c|}\hline
    Round & Player name        &season     &points             &rebound            &steals        &assists        & utility &regret ratio              \\ \hline \hline
     \multirow{3}{*}{1} & \textbf{Wilt Chamberlain}	 &1961    &4029	             &2052	             &0	            &192            &1862.7          &0\%        \\ 	
    & Michael Jordan	   &1988  &2633	             &652	             &234	        &650            &1162.3	                 &37.60\%                    \\ 						
    & Michael Jordan	   &1987  &2868	             &449                &259	        &485	        &1143.9                  &38.59\%                      \\ \hline
   \multirow{3}{*}{2} & \textbf{Wilt Chamberlain}	&1961     &4029	             &2052	             &0	            &192            &1862.7          &0\%            \\
    & Mike Conley	    &2008 &2505	             &251                &354           &276            &952.8                   &48.85\%                       \\
    & Tiny Archibald	&1972     &2719	             &223                &0	            &910            &1064.6                  &42.86\%	                \\ \hline
    \multirow{3}{*}{3} & \textbf{Wilt Chamberlain}	&1961     &4029	             &2052	             &0	            &192            &1862.7          &0\%         \\
    & John Stockton	    &1988    &1400               &248                &263	        &1118	        &770.6                   &58.63\%                       \\   						
    & Wilt Chamberlain	&1960     &3033               &2149	             &0	            &148	        &1584.2                  &14.95\%                   \\ \hline
    \multirow{3}{*}{4} & \textbf{Wilt Chamberlain}	 &1961    &4029	             &2052	             &0	            &192            &1862.7          &0\%      \\
    & Wilt Chamberlain	    &1967  &1992	             &1952	             &0	            &702            &1323.6                  &28.94\%                   \\  						
    & Isiah Thomas	       &1984  &1720	             &361	             &187	        &1123	        &886.3                   &52.42\%                       \\ \hline
    \multirow{3}{*}{5} & \textbf{Wilt Chamberlain}	&1961     &4029	             &2052	             &0	            &192            &1862.7          &0\%         \\
    & Oscar Robertson 	&1961  &2432	             &985	             &0	            &899	        &1204.9                  &35.31\%                        \\ 						
    & Michael Jordan 	&1986  &3041	             &430	             &236      	    &377            &1163.9                  &37.52\%                       \\ \hline
   \multirow{2}{*}{6} &  \textbf{Wilt Chamberlain} &1961	     &4029	             &2052	             &0	            &192            &1862.7          &0\%            \\
    & McGinnis George	     &1974   &2353	             &1126	             &206           &495	        &1183.9                  &36.44\%                   \\ \hline 						

  \end{tabular}
\end{table}
\begin{table}[!htb]
	\setlength{\belowcaptionskip}{-0.2mm}
 \centering
   	\setlength{\tabcolsep}{0.3mm}
  \caption{Sorting-Simplex example with utility function $f=<0.3,0.3,0.2,0.2>$}
  \label{tab:Sorting-Simplex}
 \footnotesize
  \begin{tabular}{|c|l|c|c|c|c|c|c|c|c|}\hline
     Round & Player name  &season &  points & rebound & steals & assists & utility & regret ratio \\\hline\hline
     \multirow{3}{*}{1} & Wilt Chamberlain \textcircled{1} & 1961 & 4029	              &2052	              &0	         &192	         &1862.7                &0\%        \\
     & Oscar Robertson \textcircled{3} &1961   & 2432	              &985	              &0             &899	         &1204.9                &35.31\%       \\ 					
     & Wilt Chamberlain \textcircled{2} &1967  & 1992	              &1952	              &0	         &702            &1323.6                &28.94\%       \\ \hline
      \multirow{2}{*}{2} & Wilt Chamberlain \textcircled{1}&1961     & 4029	              &2052	              &0	         &192	         &1862.7                &0\%        \\
     & Wilt Chamberlain \textcircled{2} &1960  & 3033	              &2149	              &0             &148	         &1584.2                &14.95\%                               \\ \hline 						
  \end{tabular}
\end{table}
\noindent
In order to describe the advantage of our Sorting-Simplex algorithm, we take the 4-dimensional NBA dataset as an example, and the four dimensions represent a play's statistics on \emph{points}, \emph{rebounds}, \emph{steals}, and \emph{assists} respectively. The method proposed in \cite{Xie:2019} named the UH-Simplex algorithm corresponds to Table \ref{tab:UH-Simplex} and our algorithm refers to Table \ref{tab:Sorting-Simplex}.  We assume the user's utility function $f$ is (0.3, 0.3, 0.2, 0.2). In the process of interaction, the maximum regret ratio between the point shown by the Sorting-Simplex algorithm and the user's favorite point is 35.31\%, and that of the UH-Simplex algorithm is 58.63\%. UH-Simplex needs 6 rounds of interaction but our Sorting-Simplex only needs two rounds. We can see that at each interaction round Wilt Chamberlain in 1961 season is with the best performance \textit{w.r.t.} the user's utility function (denoted as $p$, the user's favorite point). Even we add other players in different seasons (vertexes in $N_p$) for the user to choose, this record is still the user's favorite. For the Sorting-Simplex algorithm there are only two points displayed for the last interaction round. Since the whole candidate set $C$ only has two points left, they are both taken out for the user to choose from, and the one that the user chooses is his/her favorite point.
\vspace{-2mm}
\section{Experimental Results}
\label{sec5:exp}

\noindent
In this section, we verify the efficiency and effectiveness of our algorithms on both synthetic and real datasets.
\vspace{-5mm}
\subsection{Setup}
\vspace{-3mm}
\noindent
We conducted experiments on a 64-bit machine with 2.5GHz CPU and 8G RAM on a 64-bit whose operating system is the Ubuntu 16.04 LTS. All programs were implemented in GNU C++. The synthetic datasets were generated by the dataset generator \cite{Borzsony:2001}. The anti-correlated datasets all contains 10,000 points with 4, 5 and 6 dimensions. 
For real datasets, we adopted Island, NBA and Household datasets. Island is 2-dimensional, which contains 63,383 geographic positions \cite{Tao:2009}. 
NBA dataset\footnote{\url{https://www.rotowire.com/basketball/}} contains 21,961 points for each player/season combination from year 1946 to 2009. Four attributes are selected to represent the performance of each player, \ie \emph{total scores}, \emph{rebounds}, \emph{assists} and \emph{steals}. Household\footnote{\url{http://www.ipums.org}} is a 7-dimensional dataset consisting of 1,048,576 points, showing the economic characteristics of each family of US in 2012. All the attributes in these datasets are normalized into [0,1]. Unless specified explicitly, the number of displayed points $s$ is 4. Our algorithms were compared with previous UH-Simplex algorithm \cite{Xie:2019}, UH-Random algorithm \cite{Xie:2019}, and the UtilityApprox algorithm \cite{Nanongkai:2012}. Moreover, like studies in the literature \cite{Nanongkai:2010,Peng:2014,Nanongkai:2012,Faulkner:2015,Xie:2018}, we computed the skyline first and then identified the user's favorite point from it.
\vspace{-5mm}
\subsection{Results on Synthetic Datasets}

\noindent
In Fig. \ref{varys-synthetic}, above 5 mentioned algorithms were run on the Anti-5d dataset with the final regret ratio not more than 2\%. We varied the number of displayed points $s$ from 3 to 6 and used the number of total displayed points during the interaction to measure the performances of these 5 algorithms. In order to ensure that the user's regret ratio cannot exceed 2\%. In Fig. \ref{varys-synthetic}(a), the UtilityApprox algorithm needs to present about 112 points to the user.
When $s=3$, we find that our Sorting-Simplex algorithm finally presents only 24 points to the user, meeting the 2\% regret ratio. And the last point displayed is the user's favorite point. However, UtilityApprox needs to show 105 points and require 35 rounds of interaction to meet the requirement of the regret ratio. The UH-Simplex algorithm requires 14 rounds to meet the user's regret ratio. We observe that the Sorting-based algorithms \ie Sorting-Random and Sorting-Simplex can reduce the rounds of user interaction. Although the algorithms which exploit random point selection strategy do not provide provable guarantees on the number of interaction rounds, they are a little better than Simplex-based algorithms in rounds of interaction. Also, they need less time to execute due to their randomness (Fig. \ref{varys-synthetic}(b)). We also observe that as the number of points for each round increases, the total number of interaction rounds along with the total number of displayed points decreases. For example, when $s=3$, we need 8 rounds of interaction, showing a total of 24 points. But when $s=6$, only 3 rounds of interaction are needed, and the total number of displayed points is 18. 

\begin{figure}[htbp]
	\setlength{\abovecaptionskip}{-1.5mm}
	\setlength{\belowcaptionskip}{-1mm}
\begin{minipage}[t]{0.5\textwidth}
  \centering
    \centering
    \includegraphics[height=2.8cm]{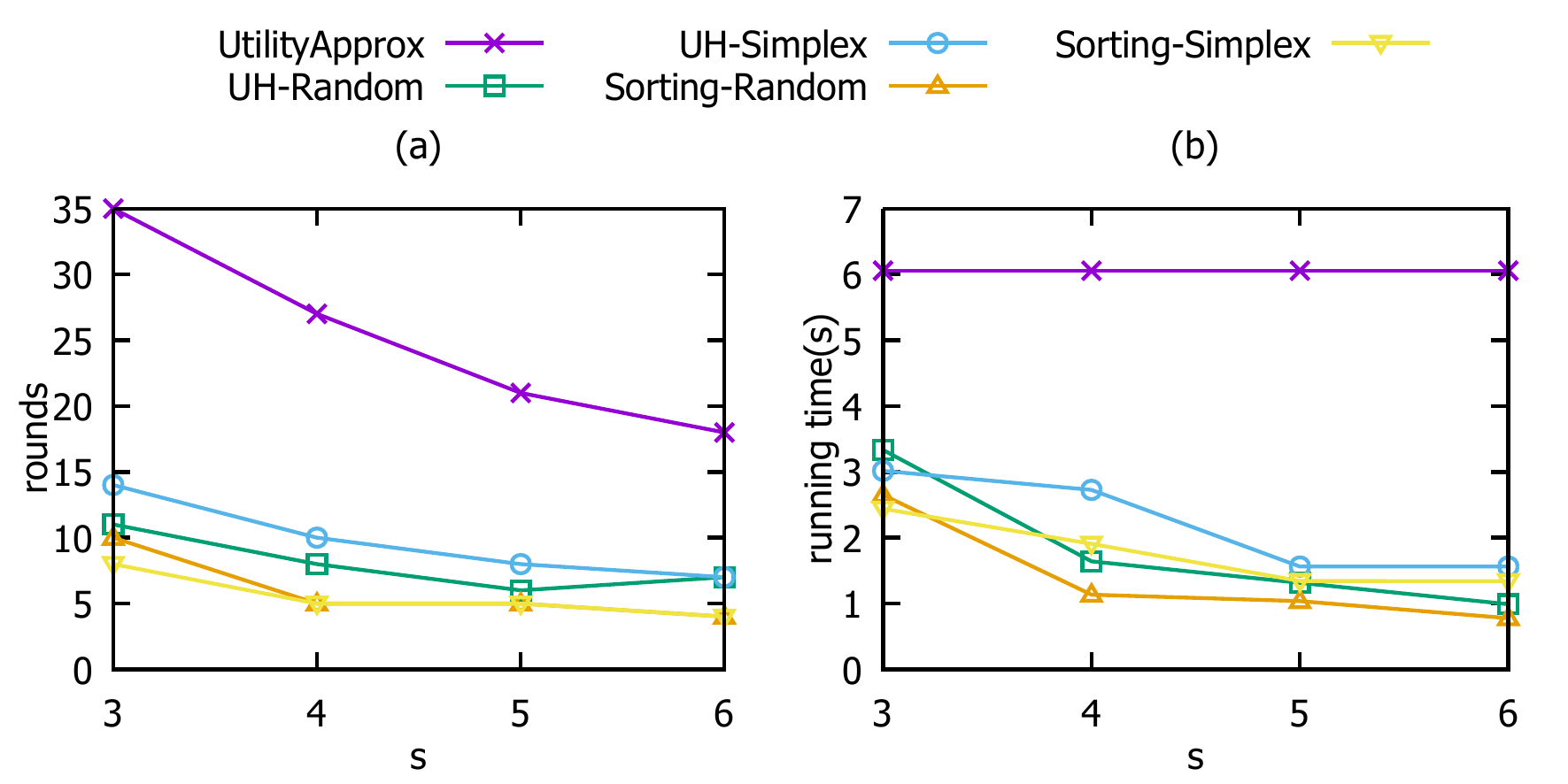}
    \caption{vary $s$ on the Anti-5d dataset}
    \label{varys-synthetic}
\end{minipage}
\begin{minipage}[t]{0.5\textwidth}
  \centering
  \includegraphics[height=2.8cm]{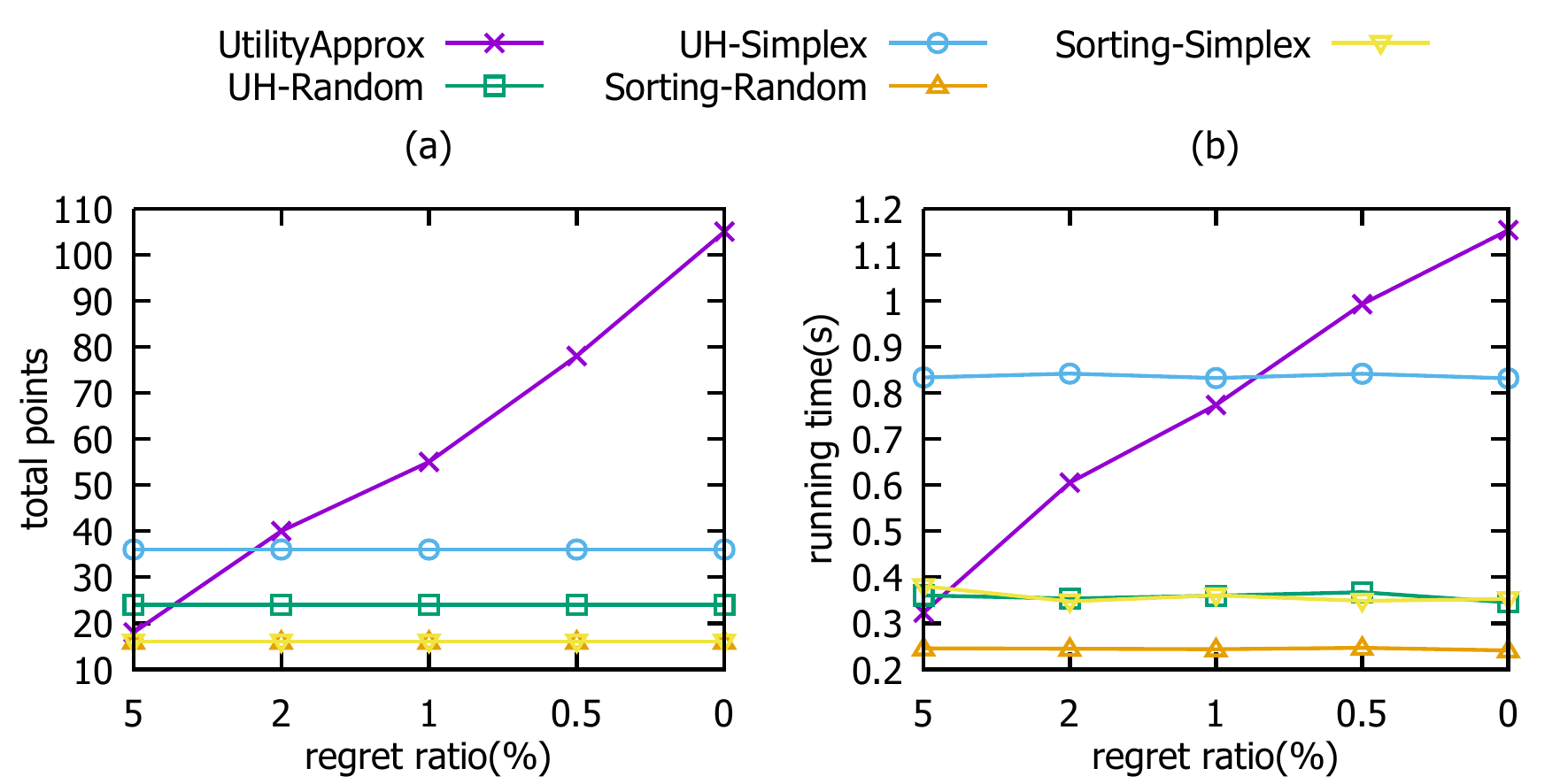}
  \caption{Vary Regret Ratio($d=4$, $s=4$, $n=10,000$)}
  \label{varyrr-fixed(n-d)}
\end{minipage}
\end{figure}

\noindent
In Fig. \ref{varyrr-fixed(n-d)}, we compared the performances of the 5 algorithms under different regret ratios. The regret ratio ranges from 5\% to 0\%, and the smaller value is better. Although we set the required regret ratio is not larger than 5\%, the regret ratios of the result sets returned by the 4 algorithms, Sorting-Simplex, Sorting-Random, UH-simplex, UH-Random are all 0\% (they are flat lines in Fig. \ref{varyrr-fixed(n-d)}(a)). But the regret ratio of UtilityApprox is 4.87\%, which performs worse than the other algorithms. We observe that the Sorting-based algorithms are better than the other algorithms, either in the number of displayed points or in the running time. And the Sorting-Simplex algorithm takes less time than UH-Simplex and UtilityApprox. The total number of displayed points of Sorting-Simplex is less than that of the UH-Simplex algorithm, because when $s$ points are shown, the UH-Simplex algorithm can only get the $s-1$ comparisons for the candidate set pruning. As a contrast, our Sorting-Simplex algorithm can get $C_s^2 $ comparisons which are exploited to delete larger amount of the data points having no possibility to be the maximum utility point from the candidate set. Also, Sorting-Simplex only needs to show half number of the points of UH-Simplex to achieve the same regret ratio. If the user wants to choose his/her favorite point, UtilityApprox needs to show 105 points compared with the other 4 algorithms. We know that the more points shown to the user, the more effort he/she will take to browse them. So UtilityApprox wastes a lot of the user's effort and takes up too much time of the user (as shown in Fig. \ref{varyrr-fixed(n-d)}(b)). This leads to the worst performance of UtilityApprox against the other 4 algorithms.

We also evaluated the scalability of our Sorting-based algorithms in Fig. \ref{varyn-time-totalpoints} and Fig. \ref{varyd-time-totalpoints}. In Fig. \ref{varyn-time-totalpoints}, we studied the scalability of each algorithm on the dataset size $n$. Our Sorting-Simplex algorithm scales well in terms of the running time while showing the smallest amount of points to the user. In particular, to guarantee a 0.1\% regret ratio on a dataset with 20,000 points, the number of points we display is half of that of UH-Simplex and one sixth of that of UtilityApprox (Fig. \ref{varyn-time-totalpoints}(a)). Besides, the other 4 interactive algorithms are significantly faster than UtilityApprox (Fig. \ref{varyn-time-totalpoints}(b)). In Fig. \ref{varyd-time-totalpoints}, we studied the scalability of each algorithm on the dimensionality $d$. Compared with UH-Simplex and UtilityApprox, Sorting-Simplex and Sorting-Random consistently show fewer points in all dimensions, verifying the usefulness of sorting points in reducing the rounds of interaction.
\begin{figure}[htbp]
	\setlength{\abovecaptionskip}{-1.5mm}
	\setlength{\belowcaptionskip}{-1mm}
\begin{minipage}[t]{0.5\textwidth}
  \centering
  \includegraphics[height=2.8cm]{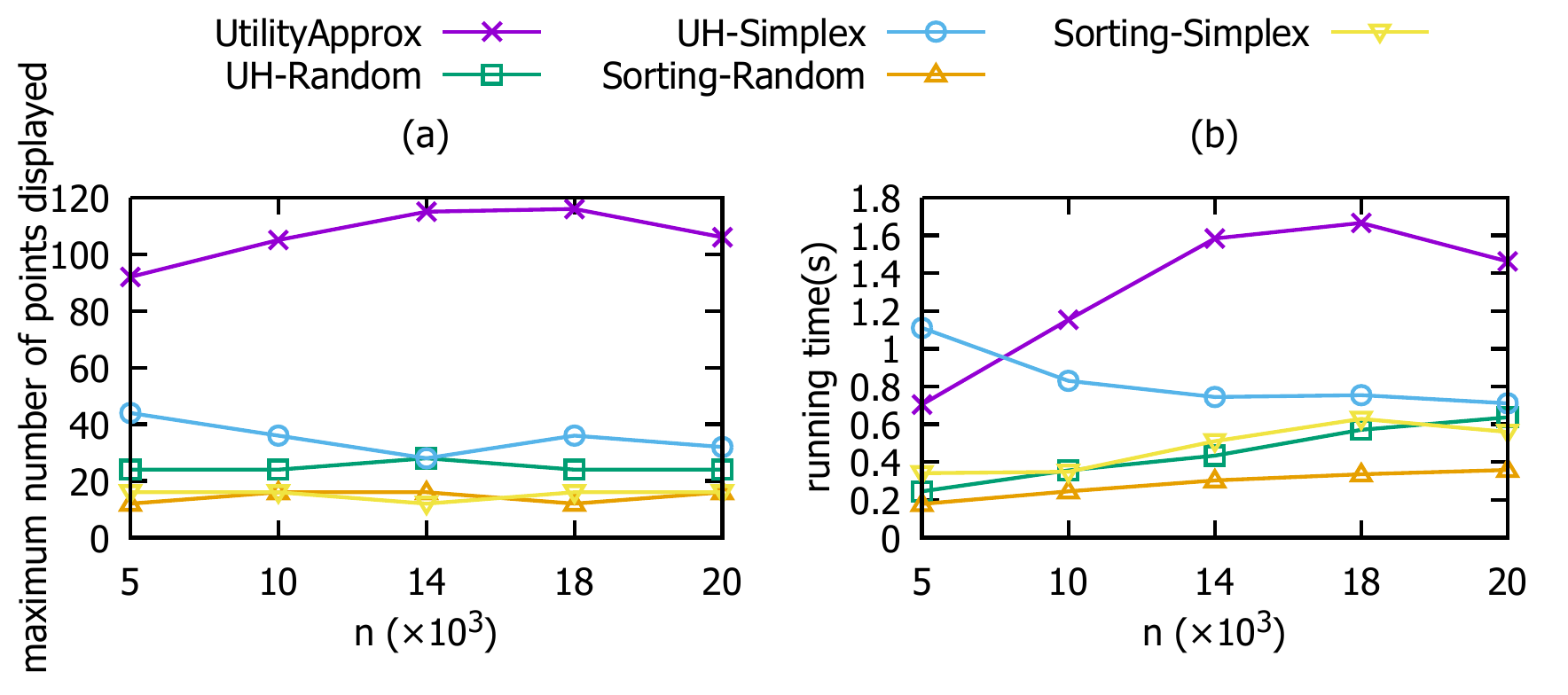}
  \caption{Vary $n$($d$=4,$s$=4,$\epsilon$=0.1\%)}
  \label{varyn-time-totalpoints}
\end{minipage}
\begin{minipage}[t]{0.5\textwidth}
    \centering
  \includegraphics[height=2.8cm]{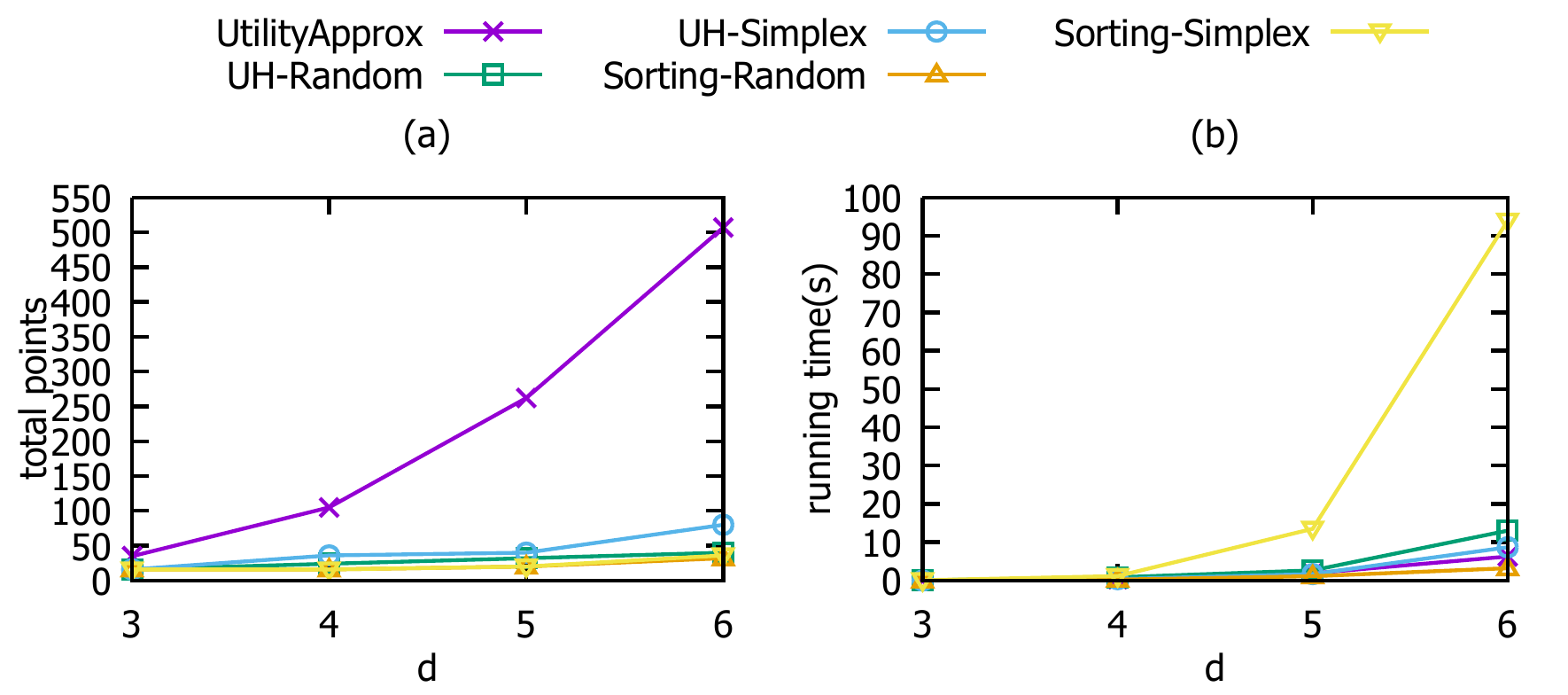}
  \caption{Vary $d$($n$=10,000,$s=4$,$\epsilon$=0.1\%)}
  \label{varyd-time-totalpoints}
\end{minipage}
\end{figure}

\vspace{-3mm}
\subsection{Results on Real Datasets}

\noindent
We studied the effects of the algorithms on the 3 real datasets in terms of the regret ratio, candidate set size and running time of each algorithm. Note that our sorting-based algorithms perform very efficiently on real datasets. This is because that sorting the displayed points can generate more information for learning user's utility function and reducing the candidate set size. Note that when the running time remains unchanged (Fig. \ref{vary-totalpoints-on-NBA}(c), Fig. \ref{vary-totalpoints-on-household}(c)), it means the points displayed in the previous interaction round satisfy the user's requirement, there is no need to present more points to the user. The random algorithms, \ie UH-Random and Sorting-Random with unstable tendency are due to the randomness for the displayed point selection. 

\begin{figure}[htbp!]
	\setlength{\abovecaptionskip}{-1.5mm}
	\centering
	\includegraphics[height=3cm]{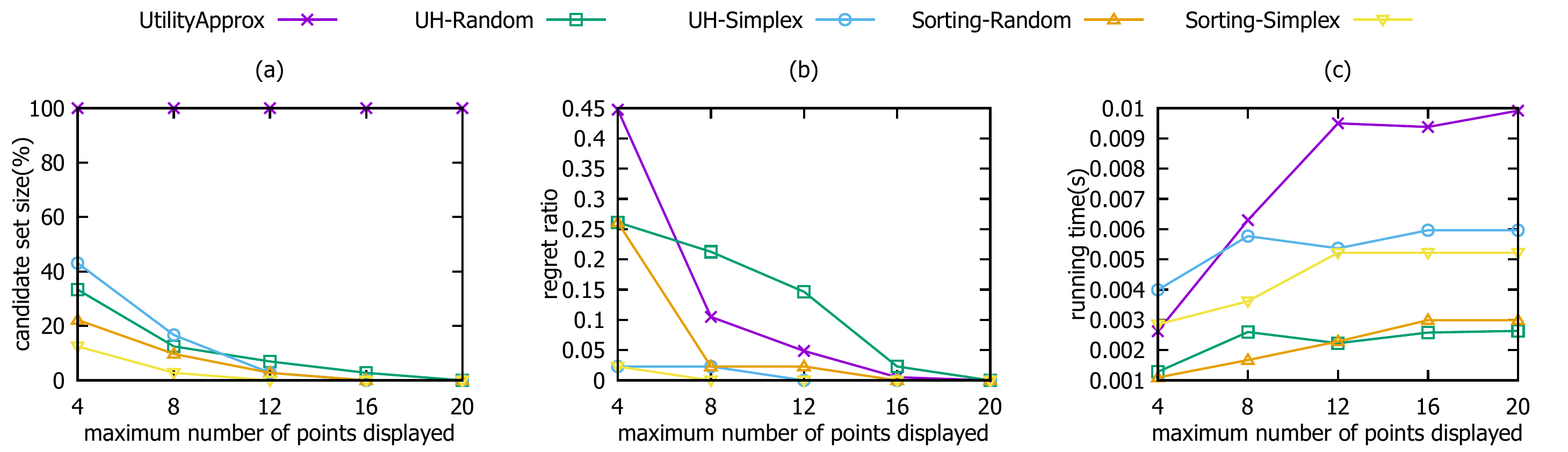}
	\caption{Vary maximum number of points displayed on NBA}
	\label{vary-totalpoints-on-NBA}
\end{figure}

\noindent
The results on the NBA and Household datasets are shown in Fig. \ref{vary-totalpoints-on-NBA} and Fig. \ref{vary-totalpoints-on-household} where we vary the maximum number of points displayed. Our sorting-based algorithms effectively reduce the candidate set size and take only a few seconds to execute. The Sorting-Simplex algorithm reached 0\% regret ratio in the 3rd round. 
When the Sorting-Simplex algorithm is executed, the candidate set size is reduced rapidly. In particular, after 2 rounds (\ie total 8 points presented to the user since $s$ = 4), we prune 98\%, 50\% of data points in the candidate set on NBA and Household as shown in Fig. \ref{vary-totalpoints-on-NBA}(a) and Fig. \ref{vary-totalpoints-on-household}(a), respectively.

\begin{figure}[htbp!]
	\setlength{\abovecaptionskip}{-1.5mm}
	\setlength{\belowcaptionskip}{-3mm}
    \centering
    \includegraphics[height=3cm]{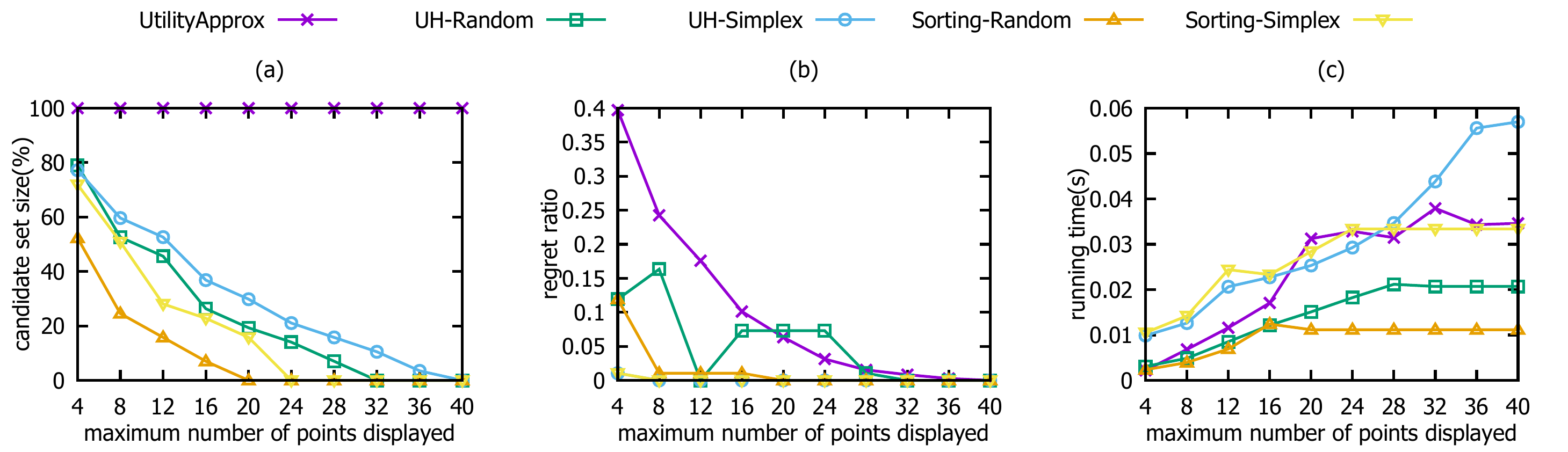}
    \caption{Vary maximum number of points displayed on household}
    \label{vary-totalpoints-on-household}
\end{figure}

\begin{figure}[htbp!]
	\setlength{\abovecaptionskip}{-1.5mm}
	\setlength{\belowcaptionskip}{-1mm}
    \centering
    \includegraphics[height=2.8cm]{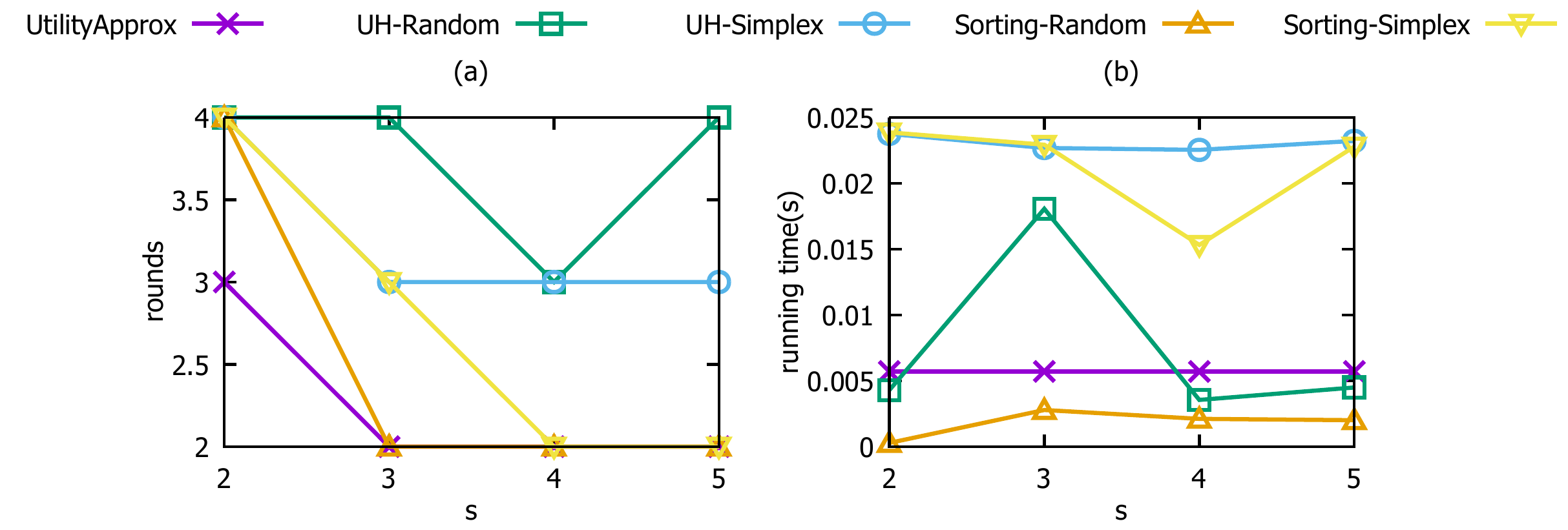}
    \caption{vary $s$ on island($d$ =2, $\epsilon =0\%$)}
    \label{varys-time-questions-on-island}
\end{figure}

\noindent
When the system required the same regret ratio of result set for each algorithm, we found that our Sorting-Simplex algorithm performs best among all the algorithms as shown in Fig. \ref{varyrr-fixed(n-d)}. Moreover for smaller target regret ratios, Sorting-Simplex clearly outperforms UH-Simplex and UtilityApprox. The same phenomenon occurs when we increase the number of points shown to the users, as shown in Fig. \ref{vary-totalpoints-on-NBA}(b), Fig. \ref{vary-totalpoints-on-household}(b). This confirms that the idea of sorting is crucial in reducing the rounds of interaction. The results on the Island dataset are shown in Fig. \ref{varys-time-questions-on-island} where we vary the number of displayed points. In Fig. \ref{varys-time-questions-on-island}(a), we find that only 3 or 4 rounds needed for interaction due to low dimensionality. From Fig. \ref{vary-totalpoints-on-NBA}(c), Fig. \ref{vary-totalpoints-on-household}(c) and Fig. \ref{varys-time-questions-on-island}(b), our sorting-based algorithms are competitive over other algorithms in running time. However, the time spent by the UtilityApprox algorithm is not longer than the UH-based algorithms due to the fact that the points presented by the UtilityApprox algorithm are artificial/fake points. These points do not take time to select from the dataset. 
\vspace{-2mm}
\section{Conclusion}
\label{sec6:conclusion}

\noindent
In this paper, we present sorting-based interactive framework for regret minimization query. With the help of nice properties of geometric objects describing multidimensional data points, such as boundary point, hyperplane, convex hull, conical hull frame, neighbouring vertex etc, we fully exploit the pairwise relationship of the sorted points to shrink the user's possible utility space greatly and reduce the size of the candidate set which has a consequence that our proposed method requires less rounds of interaction. Experiments on synthetic and real datasets verify our proposed Sorting-Random and Sorting-Simplex algorithms are superior to existing algorithms in terms of interaction rounds and running time.

\subsubsection*{Acknowledgments}
This work is partially supported by the National Natural Science Foundation of China under grants U1733112, 61702260 and the Fundamental Research Funds for the Central Universities under grant NS2020068.

\vspace{-4mm}
\bibliographystyle{splncs04}
\bibliography{ref}

\end{document}